\documentclass[11pt]{article}

\usepackage[margin=0.9in]{geometry}
\usepackage{times}
\usepackage{amsmath, amsthm}
\usepackage{cleveref}
\usepackage{listings, fancyvrb, multicol}
\usepackage{color,xcolor}
\usepackage{graphicx}

\usepackage{url}
\usepackage{authblk}
\usepackage{enumitem}
\usepackage[T1]{fontenc}
\usepackage[utf8]{inputenc}

\def\arxiv{1}

\makeatletter
\lst@Key{countblanklines}{true}[t]%
{\lstKV@SetIf{#1}\lst@ifcountblanklines}

\lst@AddToHook{OnEmptyLine}{%
	\lst@ifnumberblanklines\else%
	\lst@ifcountblanklines\else%
	\advance\c@lstnumber-\@ne\relax%
	\fi%
	\fi}
\makeatother

\newtheorem{theorem}{Theorem}[section]
\newtheorem{lemma}[theorem]{Lemma}

\newtheorem{definition}{Definition}

\title{FliT: A Library for Simple and Efficient Persistent Algorithms}

\newif\ifcomments
\commentsfalse

\ifcomments
\newcommand{\Guy}[1]{\noindent 
	{\color{red}{\textbf{Guy: }#1}}}

\newcommand{\Naama}[1]{\noindent
	{\color{blue}{\textbf{Naama: }#1}}}

\newcommand{\Hao}[1]{\noindent
	{\color{purple}{\textbf{Hao: }#1}}}
\newcommand{\hedit}[1]{{\color{purple}{#1}}}

\newcommand{\Michal}[1]{\noindent
	{\color{orange}{\textbf{Michal: }#1}}}

\else

\newcommand{\hedit}[1]{#1}

\newcommand{\Guy}[1]{}

\newcommand{\Naama}[1]{}

\newcommand{\Hao}[1]{}

\newcommand{\Michal}[1]{}

\fi

\newcommand{\FM}{FliT}

\newcommand{\fminstruction}{flit-instruction}

\newcommand{\flushed}{persisted}
\newcommand{\flushedi}{p-instruction}
\newcommand{\flushedm}{p-store}
\newcommand{\flushedl}{p-load}
\newcommand{\nonflushedi}{v-instruction}
\newcommand{\nonflushedm}{v-store}
\newcommand{\nonflushedl}{v-load}
\newcommand{\nonflushed}{volatile}

\newcommand{\flush}{pwb}
\newcommand{\fence}{pfence}
\newcommand{\pmem}{persistent memory}
\newcommand{\dram}{volatile memory}
\newcommand{\fmcounter}{flit-counter}
\newcommand{\fmadjacent}{flit-adjacent}
\newcommand{\fmhash}{flit-HT}
\newcommand{\operation}{operation}
\newcommand{\interface}{P-V  Interface}
\newcommand{\noflit}{plain}

\newif\ifpaper
\papertrue

\begin{document}






\author[1]{Yuanhao Wei}
\author[2]{Naama Ben-David}
\author[3]{Michal Friedman}
\author[1]{Guy E. Blelloch}
\author[3]{Erez Petrank}

\affil[1]{Carnegie Mellon University, USA}
\affil[2]{VMware Research, USA}
\affil[3]{Technion, Israel}

\affil[ ]{\textit{yuanhao1@cs.cmu.edu, bendavidn@vmware.com, michal.f@cs.technion.ac.il, guyb@cs.cmu.edu, erez@cs.technion.ac.il}}

  \date{}


	\maketitle 
		\begin{abstract}
Non-volatile random access memory (NVRAM) offers byte-addressable persistence at speeds comparable to DRAM. However, with caches remaining volatile, automatic cache evictions can reorder updates to memory, potentially leaving persistent memory in an inconsistent state upon a system crash. Flush and fence instructions can be used to force ordering among updates, but are expensive. This has motivated significant work studying how to write correct and efficient persistent programs for NVRAM. 

In this paper, we present FliT, a C++ library that facilitates writing efficient persistent code. Using the library's default mode makes any linearizable data structure durable with minimal changes to the code. FliT avoids many redundant flush instructions by using a novel algorithm to track dirty cache lines. 
The FliT library also allows for extra optimizations, but achieves good performance even in its default setting. 

To describe the FliT library's capabilities and guarantees, we define a persistent programming interface, called the P-V Interface, which FliT implements. The P-V Interface captures the expected behavior of code in which some instructions' effects are persisted and some are not. We show that the interface captures the desired semantics of many practical algorithms in the literature. 

We apply the FliT library to four different persistent data structures, and show that across several workloads, persistence implementations, and data structure sizes, the FliT library always improves operation throughput, by at least $2.1\times$ over a naive implementation in all but one workload. 
	\end{abstract}

	\section{Introduction}

The long-anticipated fast, byte-addressable non-volatile random access memories (NVRAM) are now a reality, with Intel Optane available alongside DRAM in the newest machines. NVRAM promises to revolutionize persistent algorithms, with speeds up to three orders of magnitude faster than SSD. However, designing correct persistent algorithms for NVRAM is notoriously difficult. Subtle bugs are easy to overlook. The main difficulty of programming for NVRAM stems from the fact that, for the time being, caches and registers remain volatile.\footnote{While Intel announced the new eADR technology~\cite{eADR} that promises to persist cache contents as well, this would require powerful and expensive batteries to implement. Thus, it is unlikely that volatile caches will cease being a reality in the near future~\cite{Scargall2020}.} This means that if programs are simply run as they would be on DRAM, significant parts of the state of memory could be lost upon a system crash, thus not achieving meaningful persistence. On the other hand, programs designed for SSD or disk cannot efficiently work as-is on NVRAM, due to the finer atomic granularity of this new memory technology. New techniques must therefore be developed to achieve correct and efficient persistence on NVRAM.

To prevent values on cache being lost upon a crash, programmers must use explicit flush and fence instructions to push cache lines to NVRAM in a certain order. Care must be taken in deciding which values to flush and when to execute the flush and fence instructions, since these instructions are expensive.
Researchers have therefore dedicated significant effort to carefully reasoning about inherent dependencies in algorithms, to omit flushes when it is safe to do so, yet still guarantee persistence on NVRAM~\cite{friedman2018persistent,friedman2020nvtraverse,ramalhete19onefile,correia2018romulus,chauhan2016nvmove,chen2015persistent,coburn2011nvheaps,cohen2018inherent,lee2017wort,lee2019recipe,xu2016nova}.

Data races in persistent programs pose even more challenges. Since writing and persisting values cannot be done atomically, a value can be visible to other threads before being persisted. Thus, to avoid memory inconsistencies, a process may have to flush locations it reads, even if processes flush locations when they write as well. 
However, in most cases, a writing process can finish persisting its new value before any other process reads it. In that case, it seems wasteful to have the reader flush this value as well.
Existing work in the literature avoids these wasteful flushes by using a bit in each memory word to indicate whether or not it has already been flushed~\cite{david2018logfree,wang2018markbit,guerraoui2020efficient}.
This optimization has been shown to have tremendous benefits in practice, but borrowing a bit from each word is not always possible. 
Furthermore, this optimization requires modifying memory using compare-and-swap, and therefore cannot be applied to data structures designed with other primitives, such as fetch-and-add or swap.

We propose a new technique for avoiding unnecessary flushes which is fully general in the sense that it can be applied to any code safely.
The idea is to use counters (separate from the memory word) to keep track of ongoing stores for each variable. 
When a store begins, it \emph{tags} the memory location it operates on by incrementing the corresponding counter. Loads check the counter when accessing a given memory location, and only execute a flush instruction on it if it is tagged. In this way, flush instructions are only executed when needed. 
\hedit{This technique allows for flexibility in the placement of these counters.	
The counters can be, for example, placed next to each variable or in a separate hash table.
We experiment with different options in Section~\ref{sec:evaluation}.}	

We package this technique into an easy-to-use C++ library called \emph{\FM}, or \emph{Flush if Tagged}, which helps programmers easily design efficient persistent code for NVRAM, abstracting away details of flush and fence instructions, and applying the optimization under the hood. 
At a high level, the \FM{} library persists the effect of each instruction without requiring the programmer to handle low-level flushing and barriers.

The \FM{} library greatly improves the performance of persistent code, since it enables the program to safely skip flush instructions when they aren't needed. Furthermore, \FM{} is easy to use, and its syntax requires minimal changes when applying it to existing code. Indeed, to use \FM{}, the programmer simply needs to modify the declaration of variables to be persisted, and annotate when an operation terminates -- this already makes any linearizable data structure durably linearizable~\cite{izraelevitz2016linearizability}. For example, a C++11 implementation of Harris's linked list~\cite{harris2001pragmatic} can be made durably linearizable using our library by changing just seven lines of code.

Another advantage of the \FM{} library is its flexibility; while, by default, the \FM{} library instruments each load and store instruction to access the tag counters, this does not have to be the case. Many previous works have focused on understanding which values must be persisted, and which can be left volatile~\cite{friedman2018persistent,friedman2020nvtraverse,david2018logfree,chauhan2016nvmove,coburn2011nvheaps,cohen2018inherent}. These efforts have led to many optimized persistent data structure implementations. The \FM{} library can complement these existing works by allowing the programmer to specify whether a specific instruction's arguments should be left volatile. In that case, the instruction can be annotated as such, and the flushing mechanism is bypassed. Thus, while the \FM{} library can be used to persist all memory values in a naive manner to yield a fairly performant solution, it can be combined with existing optimizations to yield even better results.

To more formally argue about the library's correctness, we define an \emph{abstract interface}, called the \emph{\interface}, which the \FM{} library implements. Intuitively, the interface considers two types of instructions; those whose effects must be persisted (called \emph{\flushedi{s}}), and those whose persistence has been optimized away (called \emph{\nonflushedi{s}}). The \interface{} describes the interaction between these two types of instructions and the resulting effect on the memory. We show that the \interface{} captures persistence behavior in many algorithms in the literature. Intuitively, the \interface{} abstracts flush and fence instructions down to their underlying meaning, and we use it to show that the \FM{} library behaves as expected.
We believe that the \interface{} offers a good balance between ease of programming and the efficiency of potential implementations. Since it is relatively low-level, it can be implemented efficiently, as is exemplified by \FM{}.
Furthermore, designing durably linearizable data structures~\cite{izraelevitz2016linearizability} is easy using the \interface; if every instruction is made a \flushedi, a linearizable data structure becomes durable. On the other hand, carefully reasoned optimizations can also be applied by making some instructions \nonflushedi{s} where possible. Thus, we believe that the \interface{} may be of independent interest.


We evaluate the \FM{} library by using it to implement four different durable data structures; a linked-list~\cite{harris2001pragmatic}, a BST~\cite{aravind14bst}, a skiplist~\cite{david2018logfree}, and a hash table~\cite{david2018logfree}. Furthermore, for each data structure, we evaluate three different ways of making it durable; one that makes all instructions \flushedi{s}, and two more optimized settings that appear in the literature; we consider the  NVtraverse methodology~\cite{friedman2020nvtraverse}, which allows us to have \nonflushedl{s} while traversing the data structure, and a manually optimized durable version of the same data structure~\cite{david2018logfree}. We also evaluate different settings for the placement of the counters in the implementation of \FM{}, and compare these to the existing bit-tagging technique~\cite{david2018logfree,guerraoui2020efficient,wang2018markbit}. \Guy{Might be worth mention this includes both placing adjacet to the value and using a hash table.    The idea of a hash table does not show up until later but I think it is an interesting feature of the approach.} \Hao{added a sentence about this earlier.}
We observe that, the \FM{} library provides up to $200\times$ speedup over a durable linearizable version implemented with plain flush instructions. Furthermore, even for highly optimized implementations, the \FM{} library still provides up to $4.32\times$ speedup, and never decreases the throughput of any implementation.

In summary, the contributions of the paper are as follows.

\begin{itemize}
	\item We present a new technique for tracking dirty cache lines that is fully general.
	\item We present the \FM{} library, which uses this technique to instrument instructions giving an easy way to design efficient persistent code.
	\item We formalize the \FM{} library interface as the \interface, which captures many practical use cases, and gives a simple way of creating durably linearizable data structures.
	\item We evaluate the \FM{} library and show it can significantly improve the performance of even the most optimized persistent algorithms.
\end{itemize}

The rest of this paper is organized as follows. Background is discussed in Section~\ref{sec:prelims}. In Section~\ref{sec:semantics}, we present the \interface{} definition. We rely on this interface when presenting the \FM{} library syntax in Section~\ref{sec:lib}, and its implementation and algorithmic ideas in Section~\ref{sec:algorithm}. We evaluate the \FM{} library in Section~\ref{sec:evaluation}. Finally, we discuss additional related work and conclude the paper in Sections~\ref{sec:related} and~\ref{sec:conclusion}.

	\section{Background and Preliminaries}\label{sec:prelims}


\paragraph{Flushing on current hardware} 
In existing architectures, there are specific \emph{flush} instructions which write back a value from a specific cache line to the main memory. The flush instructions might differ by their strength and whether they invalidate the cache line, which influences performance. In addition, there are \emph{fence} instructions which provide ordering. A store fence ensures that all preceding writes and flushes executed by a specific process are visible to other processes before any writes or flushes executed after the fence.
A flush followed by a fence
blocks until all previously flushed locations have reached main memory, which may be volatile (i.e., DRAM), or non-volatile (i.e., NVRAM), depending on the mapping of the specific flushed address. 

In the rest of this paper, we will use the term \emph{\flush} (persistent-write-back) to refer to the weakest form of flushing, which does not block or invalidate. It is \emph{persistent}, since in all our use-cases, the memory it flushes is mapped to NVRAM.
As mentioned above, after a non-blocking flush instruction, a fence must be called to ensure the completion of the flush. In this paper, we use \emph{\fence} to refer to a fence instruction. A \fence{} called by a process $i$ is assumed to order all previous \flush{} instructions called by $i$ before any \flush{} or write instructions that are executed after the \fence{}. 
The \flush{} and \fence{} instructions are architecture-agnostic. 


\paragraph{Previous flushing optimizations} We defer most of the discussion of related work to Section~\ref{sec:related}. However, some flushing optimizations have appeared in the literature that are reminiscent of the \FM{} library's implementation, so we briefly discuss them now. David et al.~\cite{david2018logfree} introduce a technique they call \emph{link-and-persist} to avoid executing \flush{} instructions when the variable being flushed is clean. Their technique works by using a single bit in each memory word as a flag indicating whether or not it has been flushed since the last time it was updated. When a new value is written,  it is written with the flag up. The writing process then executes a \flush{} and a \fence{} to persist the new value, and then executes another store to flip the flag down. A reader executes a \flush{} on any location it read that had the flag up, and skips flushing every time the flag is down. This technique has appeared in the literature under different names~\cite{guerraoui2020efficient,wang2018markbit, zuriel19efficient}, always optimizing redundant \flush{s}, and yielding faster algorithms. This technique is similar to the implementation of our \FM{} library. However, the \FM{} library is more general and flexible. For one, it does not require taking a bit in every memory word. While pointers leave unused bits in each word, some algorithms make use of these bits for other parts of their logic. The link-and-persist technique is not applicable to such algorithms. Furthermore, for link-and-persist to work, all stores must be executed using a CAS instruction (as opposed to, for example, \textit{fetch-and-add} or \textit{swap}), to prevent accidentally removing a flag for a value that has not yet been flushed. The \FM{} library does not suffer from these restrictions. Finally, as will be shown in the rest of the paper, the \FM{} library also provides flexibility in allocating the space used for metadata tracking persistent state, which can sometimes be a useful way to optimize implementations.


\subsection{Model}\label{sec:model}

We consider a shared memory setting in which $n$ processes access two types of memory; \dram{} and \pmem. Volatile memory roughly corresponds to caches and registers, as well as DRAM, on real architectures, whereas \pmem{} corresponds to the NVRAM. We assume that upon a \emph{system crash} anything that is in \pmem{} remains there, but anything on \dram{} is lost.

Processes can access shared memory using \emph{read}, \emph{write}, or read-modify-write (RMW) instructions, like \emph{compare-and-swap (CAS)}, \emph{fetch-and-add (FAA)}, and \emph{test-and-set (TAS)}. We sometimes refer to read instructions as \emph{loads}, and to all other instructions collectively as \emph{stores}. Each memory location is categorized at any given point in time as \emph{shared} or \emph{private}. There is a \emph{root} location in memory, that is always shared. A private memory location can only be accessed by a single process $i$,  which can make that location shared by executing a specific store on some shared location. This store depends on the algorithm; it could be releasing some lock, or swinging a shared pointer  to point to this location.

All accesses are applied to \dram. To make a value that is in \dram{} appear in \pmem, processes can execute \emph{persistence instructions}, which include \emph{\flush} and \emph{\fence} (as long as these addresses are mapped to the persistent memory). From now on in the paper, whenever we invoke \flush{} on some location, we assume that location is mapped to \pmem. A \flush{} instruction takes a memory location as a parameter. The value $v$ in memory location $\ell$ is said  to be \emph{flushed} if a \flush{} instruction was executed on $\ell$ when $v$ was in $\ell$. After a process $i$ executes a \fence{} instruction, any value that was flushed by $i$ is in \pmem.

A data structure $D$ defines a set of operations, along with a \emph{sequential specification}, defining how the operations behave in a sequential execution.
Histories are composed of \emph{operations} and \emph{crash events}. 
A crash event erases all values in \dram, but leaves the \pmem{} intact. Furthermore, after a crash event, new processes are spawned.
A history $H$ of operations of data structure $D$, with no crash events, is \emph{linearizable} if there is a single point in time during the execution of each operation at which that operation \emph{takes effect}, such that the sequence of these points adheres to the sequential specification of $D$. A history with crash events is \emph{durably linearizable}~\cite{izraelevitz2016linearizability} if it is linearizable after all crash events are removed from it.
A data structure implementation is linearizable (resp. durably linearizable) if all possible histories of it are linearizable (resp. durably linearizable).

	\section{Persistent-Volatile Instruction Interface}\label{sec:semantics}

Before presenting the \FM{} library, we define the abstract interface that it implements. This interface, called the \interface{}, is important for discussing the correctness of the \FM{} library implementation; we later prove that our implementation satisfies the abstract interface. Furthermore, this interface allows users of the \FM{} library to reason about their code in a precise manner. 
	
The \interface{} aims to capture the behavior of a program with both volatile and persistent memory. 
Firstly, handling persistence should not affect the behavior of the volatile memory. In particular, this means that we should expect to see the same sequential semantics on volatile memory as we do in a classic system. That is, any load on volatile memory should return the value written by the most recent store. 
For persistence, we expect the interface to capture the behavior of code that uses \flush{} and \fence{} instructions. We note also that dependencies between instructions can play a role in when we expect a value to be persisted; if a value has been written but never read, it may be ok for it to be lost upon a system crash, since its effects have not yet been observed. We formalize these intuitions below.

We begin defining the interface by introducing terminology to separate two types of instructions: we say an instruction is a \emph{\flushedi} if it has to be persisted (defined below), and a \emph{\nonflushedi} if it does not. More specifically, we refer to \flushed{} loads and stores as \emph{\flushedl{s}} and \emph{\flushedm{s}} respectively, and to their \nonflushed{} counterparts as \emph{\nonflushedl{s}} and \emph{\nonflushedm{s}}. If we do not specify whether an instruction is \flushed{} or \nonflushed{}, then it could be either. 

To nail this down precisely, we further distinguish between \emph{shared instructions}, which operate on shared memory, and \emph{private instructions}, which operate on a private memory location. A private instruction may allow more flexibility in when it is persisted, since other processes cannot observe its effects. Moreover, private values are not needed after a crash, since we assume new processes are spawned.

We refer to the memory location an instruction operates on as its \emph{location}. Furthermore, we associate a \emph{value} with each instruction; a \emph{load}'s value is what it returned (read from its location), and a \emph{store}'s value is the value newly written on its location. When we say an instruction is \emph{persisted}, we mean its value is on persistent memory.

To create durable code, we must reason about \emph{dependencies} among different instructions. In particular, for a new store to be safe in a persistent setting, a process $i$ must ensure that all its dependencies have been persisted \emph{before} executing the store. That is, the values that process $i$ used to determine the value and location of the new store must not be lost at a later time. Furthermore, to maintain a store-order guarantee for persistent memory, previous store instructions by the same process $i$ must also be persisted before $i$'s new store. Finally, to prevent losing the effects of a completed operation, we must persist all of $i$'s dependencies and store values before $i$ completes an operation.

The \interface{}, defined in Definition~\ref{def:spec}, formalizes the meaning of dependencies in terms of \flushedm{s} and \flushedl{s}; a process $i$ depends on its own \flushedm{s} (Condition~\ref{cond:persistStore}), and on previous \flushedm{s} on locations on which $i$ executes a \flushedl{} (Condition~\ref{cond:persistLoad}). The interface then requires that these dependencies be persisted before $i$ executes a store that is visible to other processes (shared), or before it completes an operation (Condition~\ref{cond:fence}). To capture which \flushedm{s} become dependencies, we consider the \emph{linearization} of instructions. Intuitively, an instruction linearizes at the time it accesses volatile memory (Condition~\ref{cond:seqSpec}). Note that Conditions~\ref{cond:seqSpec},~\ref{cond:persistStore}, and~\ref{cond:persistLoad} apply to both private and shared instructions, and that \nonflushedi{s} don't add dependencies.





\begin{definition}\label{def:spec}[The \interface.]
	Each instruction has a linearization point within its interval, such that:
	\begin{enumerate}
		\item \label{cond:seqSpec} \textbf{Keeping Volatile Memory Behavior.} 
		A load $r$ on location $\ell$ returns the value of the most recent store on $\ell$ that linearized before $r$.
		
		\item \label{cond:persistStore} \textbf{Store Dependencies.} Let $s$ be a linearized \flushedm{} executed by a process $i$. $i$ depends on $s$.  
		
		\item \label{cond:persistLoad} \textbf{Load Dependencies.} 
		Let $r$ be a \flushedl{} by process $i$ on location $\ell$. $i$ depends on every \flushedm{} on $\ell$ that was linearized before $r$.
		
		\item \label{cond:fence} \textbf{Persisting Dependencies.} 
		Let $t$ be either the linearization point of a shared store by process $i$, or the time at which $i$ completes an \operation. The value of every store $i$ depended on before time $t$ is persisted by time~$t$. 
	\end{enumerate}
\end{definition}

\subsection{Applicability of the \interface}\label{sec:applications}
In this subsection, we show that for many algorithms designed for NVRAM, it is easy to replace their memory accesses and all \flush{} and \fence{} instructions with \flushedi{s} (for dependencies) and \nonflushedi{s} (for instructions optimized out as non-dependencies). 

\paragraph{Simple durability} 
We begin by considering how to guarantee durability using the \interface{} for any given linearizable algorithm. 
Izraelevitz et al.~\cite{izraelevitz2016linearizability} show that, for any linearizable data structure, if every load-acquire and store-release is accompanied by a \flush{} and a \fence{}, and stores are followed by a \flush{}, then the data structure becomes durable. We show that declaring these instructions as \flushedi{s} achieves the same guarantee. Furthermore, using our implementation of the \interface{} yields a much faster solution.

\begin{theorem}\label{thm:iz-semantics}
	Given a linearizable data structure, if we make all its loads and stores \flushedi{s}, then the resulting data structure is durably linearizable.
\end{theorem}

%
%
Recall that durable linearizability requires that a history be linearizable after we remove all crash events from it. Intuitively, this means that the state of memory after each crash should be `explainable' by a feasible history without a crash. This would hold if there haven't been any `inversions' between \dram{} and \pmem{} in a way that isn't allowed  by the operations. That is, a bad situation would be if some memory update $m$ had occurred, some process $i$ observed $m$, and made a modification $m'$ that depended on having seen value $m$. If $m'$ became persisted before $m$, and then a crash occurred, the state of memory following the crash would not be explainable by any feasible history, since $i$ would never have written $m'$ had it not read~$m$.

Loosely speaking, to prove Theorem~\ref{thm:iz-semantics}, we therefore have to show that no memory update  $m$ that has not yet been persisted could be followed by another update $m'$ that depended on $m$. Intuitively, the \interface{} can guarantee this; note that Condition~\ref{cond:fence} forces all values read by some process $i$ through a \flushedl{} to be persisted before the next \flushedm{} by $i$ (since Condition~\ref{cond:persistLoad} guarantees that $i$ would depend on all such values). Due to lack of space, the full proof is deferred to the supplementary material.

\paragraph{NVTraverse} While the above construction is very simple, and can be easily applied  to any linearizable algorithm to make it persistent, there may be opportunities to optimize such a construction if some instructions could identified as non-dependencies (marked as \nonflushedi{s}). This can give more flexibility to the underlying implementation to omit \flush{} and \fence{} instructions where possible.
Indeed, there are several constructions of durable data structures in the literature that do not persist every memory instruction. For example, Friedman et al.~\cite{friedman2020nvtraverse} present a general construction to make certain lock-free data structures persistent more efficiently than the construction of Izraelevitz et al. mentioned above.  In particular, they consider data structures in \emph{traversal form}, in which each operation has a read-only traversal phase followed by a short critical phase. Many lock-free data structures, including linked-lists, BSTs, and skip lists, can fit  this form. Friedman et al. show that such data structures do not need to execute any \flush{} instructions during the traversal phase. That is, any load in the traversal phase can be thought of as a \nonflushedl{}, and any instruction (load or store) in the critical phase can be thought of as a \flushedi{}. There is a short \emph{transition} between the traversal and critical phases in NVtraverse, in which some locations that were read during the traversals are flushed. This can be achieved by executing \flushedl{s} on those locations.

\paragraph{Other Algorithms} Many other NVRAM algorithms appear in the literature, with various techniques to optimize the interaction with \pmem. As a general rule of thumb, any instruction that isn't immediately followed by a \flush{} in such algorithms can be seen as a \nonflushedi{}, and any other instruction can be seen as a \flushedi. The `dependency' terminology is used intuitively in several works~\cite{friedman2018persistent,david2018logfree}; generally, non-dependencies in those works can be seen as \nonflushedi{s}.
	\section{The \FM{} Library and Interface}\label{sec:lib}

In this section, we introduce the \FM{} library, which implements the \interface{} defined in Section~\ref{sec:semantics}. 
At its core, \FM{} provides an interface with which to declare each instruction as either a p- or \nonflushedi{} \hedit{(using the \lstinline{pflag} parameter)}.



The \FM{} library is implemented in C++. To use the library, a programmer must declare variables as \lstinline{persist<>}. 
The \lstinline{persist} template can take any type. Declaring a variable in this way essentially allows the \FM{} library to track its persistence state. Whenever this variable is accessed for loads or stores, the instruction is overloaded with the library's implementation of it, which we call a \fminstruction. Each \fminstruction{} takes the standard arguments for its underlying instruction, in addition to a flag specifying whether it is a v- or a p-instruction. Finally, a special \lstinline{operation_completion} function is made available, which must be called at the end of each data structure operation. We show the basic interface in Figure~\ref{alg:interface}.

\begin{figure}
\begin{lstlisting}[keywords=none,basicstyle=\scriptsize\ttfamily]
class persist<T> {
public member functions:
	T load(bool pflag);
	void write(T value, bool pflag);
	bool CAS(T oldval, T newval, bool pflag);
	T exchange(T newVal);
	int FAA(int amount, bool pflag);
	// FAA is only supported if T is an int type	
public static functions:
	void operation_completion(); } ;
\end{lstlisting}
\caption{Basic interface of \FM. 
}
\label{alg:interface}
\end{figure}

The \FM{} library further improves the syntax of this interface to allow for minimal code changes to apply  it. In particular, when declaring a variable in the \lstinline{persist} template, a default \lstinline{pflag} value can be specified, making the \lstinline{pflag} argument optional when executing instructions on this variable. Furthermore, for C++, we overload the \lstinline{->} and \lstinline{=} operators to execute \FM{} loads and stores instead of the default one. These operators can only be used with the default \lstinline{pflag} value though, since it does not allow for an additional argument.

Algorithm~\ref{alg:library} shows an example of the implementation of a concurrent binary tree, achieving durability by making all instructions \flushedi{s}. The change over the original code is highlighted in red. All fields within a node are declared with the \lstinline{persist<>} template, and given the \flushed{} option as a default for the \lstinline{pflag}. This means that without any code changes, all accesses to these node fields will be \flushed{} \fminstruction{s}. 
In the example, all code elided inside the `$\ldots$' remains identical to the original implementation.


\renewcommand{\figurename}{Algorithm}
\begin{figure}
	\caption{\FM{} library used for a concurrent BST.}
\begin{lstlisting}[columns=fullflexible,breaklines=true,basicstyle=\scriptsize\ttfamily]
struct Node {
	@{\color{red}{persist<}}int\color{red}{, flush\_option::\flushed>}@ key;
	@{\color{red}{persist<}}T\color{red}{, flush\_option::\flushed>}@ value;
	@{\color{red}{persist<}}std::atomic<Node*>\color{red}{, flush\_option::\flushed>}@ right;
	@{\color{red}{persist<}}std::atomic<Node*>\color{red}{, flush\_option::\flushed>}@ left;	};

@{\color{red}{persist<}}Node*\color{red}{>}@ root;

// automatic BST lookup
void lookup(int key) {
	Node* node = root->left;
	while(node->left != nullptr) {
		if(key < node->key) node = node->left;
		else node = node->right;	}
	bool result = (node->key == key);
	@{\color{red}{persist::operation\_completion();}@
	return result;	}

// automatic BST insert
bool insert(K key, V val) {
	...
	@{\color{red}{persist::operation\_completion();}@
	return result;	}
\end{lstlisting}
\label{alg:library}
\end{figure}
\renewcommand{\figurename}{Figure}

The example above only shows the use of a single setting; all instructions are called as the default \flushedi{s}. The \FM{} library is in fact more flexible, and is still easy to use even for more complicated code. Figure~\ref{fig:maunal-compare} shows a comparison between the code of a concurrent durable BST optimized to minimize \flush{} and \fence{} instructions, with the \FM{} implementation of the same code. To save space, we only show the \lstinline{insert} operation, and elide all identical code  between the two versions. This optimized version of the durable BST does not execute any persistence instructions following most memory accesses. We therefore found that setting the default \lstinline{pflag} option to be a \nonflushedi{} yields the least amount of code change. 
We note that while not shown in these examples, it is also possible to leave a variable declaration as-is, without using the \lstinline{persist} template, if that variable never requires persistence. This use case arises in some algorithms. For example, Friedman et al.~\cite{friedman2018persistent} present a durable queue implementation that completely avoids flushing the head and tail pointers of the queue. In this case, these variables can be declared normally, without the \FM{} library.

\begin{figure*}
\begin{minipage}{.45\textwidth}
\begin{lstlisting}[columns=fullflexible,breaklines=true,basicstyle=\scriptsize\ttfamily]
Node* root; // root of manual BST without @\FM@
bool insert(K key, V val) {
	...
	PWB(&node->left);
	...
	node->right.CAS(exp, new);
	PWB(node->right);
	...
	PFENCE();
	return result;	}
\end{lstlisting}
\end{minipage}
\hspace{0.2cm}
\begin{minipage}{.48\textwidth}
\begin{lstlisting}[columns=fullflexible,breaklines=true,basicstyle=\scriptsize\ttfamily]
// manual BST with @\FM@
persist<Node*, flush_option::@\nonflushed@> root;
bool insert(K key, V val) {
	...
	node->left.load(); // to flush persistent stores
	...
	node->right.CAS(exp, new, flush_option::persisted);
	...
	persist::operation_completion();
	return result;	}
\end{lstlisting}
\end{minipage}
\caption{Comparison of a manually persistent BST implemented with and without the \FM{} library.}
\label{fig:maunal-compare}
\end{figure*}

\section{The Algorithm}\label{sec:algorithm}
We now describe the implementation of the \FM{} library, and prove that it satisfies the \interface{} specified in Section~\ref{sec:semantics}.
At a high-level, each \flushedm{} \fminstruction{} executes a \fence{} before its store, and a \flush{} on this location after the store. This means that already, conditions \ref{cond:seqSpec}, \ref{cond:persistStore} and \ref{cond:fence} are satisfied (\hedit{ignoring the dependencies from Condition \ref{cond:persistLoad}}). If we had a guarantee that every \flushedl{} to any location $\ell$ will always happen after persisting the most recent \flushedm{} on $\ell$, then Condition~\ref{cond:persistLoad} would be satisfied as well, without having to change the implementation of load instructions at all. However, this is not the case; since we cannot store and persist atomically, it is possible for another process to read a value written into $\ell$ by a shared \flushedm{} before the writing process persists. 

One way \hedit{to handle dependencies from} Condition~\ref{cond:persistLoad} is to have each \flushedl{} execute a \flush{} after reading its value. However, this would introduce many unnecessary \flush{s}, since most \flush{s} do not execute concurrently with a pending \flushedm{} on the same location. Our goal in this work is to avoid as much excessive flushing as possible.

The basic idea behind the implementation of \FM{} is to associate each \lstinline{persist} variable with a counter, which we call the \emph{\fmcounter}. Intuitively, this counter keeps track of the number of pending \flushedm{} \fminstruction{s}. When a \flushedm{} \fminstruction{} begins, it increments its associated \fmcounter. It then executes its modification, followed by a \flush{} on this location, and then decrements the counter. This counter is checked by all \flushedl{s} on this location, and if its value is non-zero, the \flushedl{} executes a \flush{} after reading the value. 
A location whose \fmcounter{} is non-zero is said to be \emph{tagged}, and \flushedl{s} only flush locations that are tagged (i.e. Flush if Tagged (FliT)); this is where the \FM{} library gets its name.

Store \fminstruction{s} also execute a \fence{} before beginning their execution, and another one before decrementing the counter (in the case of a \flushedm). These \fence{s} ensure that all modifications are persisted at the correct times according to Definition~\ref{def:spec}. In particular, the \fence{} before a store ensures Condition~\ref{cond:fence} holds, by making sure all values \hedit{\flush{ed} by this process (which includes all of its dependencies) have been persisted.} The \fence{} before decrementing the counter is required for Condition~\ref{cond:persistLoad}; if this \fence{} is not executed, a \flushedl{} may observe the \fmcounter{} at value $0$ and avoid flushing the location, even though the written value has not yet been persisted.

The \FM{} library implementation distinguished between shared and private accesses to the memory. The details above in fact describe the implementation for shared accesses.  If a given \fminstruction{} is private, then its implementation is more efficient; we can ignore the \fmcounter{} associated with the accessed location, and avoid the \fence{} before a private \flushedm. Intuitively, if a \fminstruction{} cannot be concurrent with any other, then the accessed location is guaranteed not to be tagged (i.e. its counter has value 0), and return to this state (in the case of a store) before the next \fminstruction{} accesses it. Therefore, there is no need to check it, or to leave any traces for other processes. Furthermore, note that Condition~\ref{cond:fence} of the \interface{} only requires persisting before \emph{shared} stores, so we can skip the \fence{} before private stores. Unless specified otherwise, we always discuss shared \fminstruction{s} in the text, since their implementation is more involved than that of private ones.
The pseudocode of the implementation of instructions on \lstinline{persist} variables is presented in Algorithm~\ref{alg:fmalg}. Recall that p- and v-instructions are distinguished by the \lstinline{pflag} argument. We combine all types of store \fminstruction{s} (CAS, FAA, write, etc) into one in the pseudocode, since their behavior is the same.

Note that we do not specify how each memory location is associated with a \fmcounter. In Section~\ref{sec:counter}, we discuss possible ways to assign counters to memory locations, but we note that this is flexible. In particular, having many concurrent stores to the same memory location, or sharing a \fmcounter{} among several locations, cannot result in unsafe behavior (though it may result in extra \flush{s} executed). 



\renewcommand{\figurename}{Algorithm}
\begin{figure}
	\caption{The Flush-Marking Algorithm}
\begin{lstlisting}[basicstyle=\scriptsize\ttfamily]
T shared-load(T* X, bool pflag) {
	T val = X.load();
	if (pflag) {
		if (@\fmcounter@(X) > 0) { PWB(X);}	}
	return val;	}  

T private-load(T* X, bool pflag){
	return X.load();	}  

void shared-store(T* X, T* args, bool pflag) {
	PFENCE();
	if (pflag) {
		@\fmcounter@(X).fetch&add(1);
		X.store(args); @\label{line:pstore}@
		PWB(X);    @\label{line:flushstore}@  
		PFENCE(); @\label{line:fence2}@
		@\fmcounter@(X).fetch&sub(1);      	}
	else X.store(args);	}

void private-store(T* X, T* args, bool pflag){
	if (pflag){
		X.store(args); @\label{line:nr-pstore}@
		PWB(); @\label{line:nr-flushstore}@  
		PFENCE(); @\label{line:nr-fence2}@
	} else X.store(args);	}

void completeOp(){
	PFENCE();	}
\end{lstlisting}
\label{alg:fmalg}
\end{figure}
\renewcommand{\figurename}{Figure}

We now argue that our implementation (Algorithm~\ref{alg:fmalg}) satisfies Definition~\ref{def:spec}. We begin with a simple lemma.

\begin{lemma}\label{lem:nonneg}
	The value  of any \fmcounter{} is always non-negative.
\end{lemma}

\begin{proof}
	Only \flushedm{s} change the value  of an \fmcounter. Furthermore, each \flushedm{} increments the \fmcounter{} associated with its location exactly once, and decrements the same counter exactly once. The increment is always executed before the decrement. Therefore, the balance on the \fmcounter{} after a \flushedm{} terminates is  always $0$, and the balance during a \flushedm{} is either $0$ or  $1$.
\end{proof}

\begin{theorem}
	Algorithm~\ref{alg:fmalg} satisfies Definition~\ref{def:spec}.
\end{theorem}

\begin{proof}
	We let the linearization point of each \fminstruction{} on a location $X$ be the time at which it accesses the volatile memory at $X$. To argue that dependencies are persisted by the time Condition~\ref{cond:fence} dictates, we argue that a \flush{} is executed on all dependencies defined in Conditions~\ref{cond:persistStore} and~\ref{cond:persistLoad}, and that a \fence{} is executed after these \flush{s} and before the time at which they need to be persisted. Given this approach, we handle each condition separately.
	
	\textsc{Condition~\ref{cond:seqSpec}.} Note that all load (resp. store) \fminstruction{s} execute a single load (resp. store) instruction, and the arguments/return values of thetherefore \fminstruction{} are the same as the atomic instruction that it executes. 
	
	\textsc{Condition~\ref{cond:persistStore}.} Each \flushedm{} executes a \flush{} instruction (Line~\ref{line:flushstore} or~\ref{line:nr-flushstore}),  followed by a \fence{} instruction (Line~\ref{line:fence2} or~\ref{line:nr-fence2}), both of which are after its atomic store instruction (Line~\ref{line:pstore} or~\ref{line:nr-pstore}). Therefore, the value stored by this \fminstruction{} is \flush{ed} and \fence{d} before the \fminstruction{}  terminates,  and in particular has been \flush{ed} before the executing process executes another \flushedm{} or completes an operation.
	
	\textsc{Condition~\ref{cond:persistLoad}.} Consider a \flushedl{}, $r$ on location $X$, executed by process $i$. We consider two cases: either $r$ executes a \flush, or it does not. If $r$ executes a \flush{} on $X$, then, since a \flushedm{} linearizes when it is applied to the \dram{}, all dependencies on this location get \flush{ed}, and we are done. 
	If $r$ does not execute a \flush, then we again split into two cases. If $r$ is a shared-load, by the algorithm, $r$ must have read a non-positive value in $X$'s \fmcounter. By Lemma~\ref{lem:nonneg}, this means $r$ read the value  $0$ on the \fmcounter. Consider any shared \flushedm{} \fminstruction{}, $s$, on $X$. $s$'s  store linearizes after its increment of  the \fmcounter{},  and $s$ \flush{s} $X$ after its store linearizes, and executes a \fence{}, before decrementing the \fmcounter. Therefore, if $s$'s store has linearized and the \fmcounter{'s} value is $0$, then $s$'s value must already be persisted, so the condition holds. Furthermore, any private \flushedm{}, $s'$, on $X$, must have completed its execution, executing a \flush{} and a \fence{}, before $r$ could access it. Finally, if $r$ is private, then no store could be pending while $r$ executes. In particular, this means that all values stored by a \flushedm{} on this location are already persisted, by the argument above.
	
	\textsc{Condition~\ref{cond:fence}.} Note that every shared store \fminstruction{} executes a \fence{} before any other instruction. Therefore, a process $i$ executing a shared store \fminstruction{} $s$ ensures that all values on which it executed a \flush{} before beginning $s$ are persisted before $s$ linearizes. By the arguments above, this includes all dependencies of $i$, unless they have already been persisted earlier. Similarly, $i$ also executes a \fence{} at the end of each operation, when \lstinline{completeOp()} is called.
\end{proof}


\subsection{Placement of the Counter}\label{sec:counter}


In Algorithm~\ref{alg:fmalg}, we intentionally abstracted away how \fmcounter{s} are assigned to memory locations, using the unspecified \lstinline{flit-counter()}  function. Note that the \fmcounter{} is completely decoupled from the memory locations it represents, so it can be placed anywhere, and can be shared by any number of locations. Furthermore, the \fmcounter{s} can be very small; the maximum value in a \fmcounter{} is at most the number of concurrent processes in the system, since each process can increment at most on \fmcounter{} at most once before decrementing it. Therefore, on most machines, including the one we test on, $8$ bits suffice to store a \fmcounter{} without the possibility of overflow. 

In this section, we discuss a couple of practical implementations for the  \lstinline{flit-counter()} function, which we later implement and test. However, we remind the reader that other practical implementations are possible, and that the \FM{} library allows the flexibility of modifying the counter placement to suit the needs of the user.

\paragraph{Adjacent Counter.} One straightforward way to implement the \fmcounter{}s is to place each counter adjacent to the memory word that uses it. That is, we can make each memory word in an algorithm be a double-word, and use the second word for the \fmcounter. The advantage of this is that the counter for each word $X$ is on the same cache line as $X$, and therefore accessing it has minimal cost. However, this approach can be inconvenient and wasteful, since this fundamentally changes the memory layout of a given data structure's objects. Indeed, an object that fit in a single cache line might overflow it if all its fields double in size.

\paragraph{Hashed Counter.} Another \fmcounter{} placement strategy is to use a hash table; each memory location $X$ hashes into the table, which has a counter in each of its entries. This method allows different memory locations to use the same \fmcounter{}. 
The number of collisions depends on the ratio of the size of the hash table and the number of threads in the system, since each thread can access at most one hash-table entry per \fminstruction.
The advantage of this approach is two-fold. First, it saves memory. In many data structures, especially if they are not highly-contended, most memory locations will have no pending \flushedm{s} most of the time. This means that sharing counters results in a negligible amount of extra flushing. Secondly, it does not require changing the layout of memory in the data structure itself, since the \fmcounter{s} are not placed in the same cache lines as the data structure elements. However, this can also be a downside in some situations; since the \fmcounter{} is in a separate cache line, accessing it could incur an additional cache miss. 

Note also that the hashing method allows us to compact the memory usage of \fmcounter{s} even further, by squeezing several counters into each word. Recall that 8 bits suffice for each \fmcounter{}, so we can fit 8 counters in a single memory word. However, compacting the \fmcounter{s} in this way can increase false-sharing; many different memory locations could be mapped to counters on the same cache line. 
	\section{Evaluation}\label{sec:evaluation}

The \FM{} library's implementation optimizes \flush{} instructions on shared locations. To highlight its effects and focus on them in the evaluation, we evaluate the library applied to lock-free data  structures, in which most memory accesses are shared. 
We apply the \FM{} library to 4 lock-free data structures; a linked-list~\cite{harris2001pragmatic}, a binary search tree (BST)~\cite{aravind14bst}, a skiplist~\cite{fraser2003practical}, and a hash table which uses Harris's linked list to implement each bucket~\cite{harris2001pragmatic}. For each data structure, we implement three different ways of making it durable; the first is the \textbf{automatic} transformation discussed in Section~\ref{sec:applications}, in which all instructions are made \flushedi{s}, the second is using the \textbf{NVtraverse} framework~\cite{friedman2020nvtraverse}, and the third is a hand-tuned (\textbf{manual}) construction based on algorithms presented by David et al~\cite{david2018logfree}. 
We also study  the effect of various policies for placing the \fmcounter{s} in memory with respect to the memory locations they are associated  with. In particular, we implement the adjacent counter variant (\fmadjacent) and a hash table (\fmhash), for which we test five different sizes. We evaluate the tradeoffs of the different approaches.
Finally, we also implement the link-and-persist technique on the data structures that can support it. We compare these implementations of the \interface{} with the \textbf{\noflit{}} version, which places \flush{} and \fence{} instructions where necessary, but does not utilize any tagging method to avoid \flush{s} in the loads.   \Guy{Not sure why it is called an ad-hoc method?} \Hao{I agree ad-hoc is a weird name, but I'm not sure what else we can call it.}

\subsection{Setup.}

We run experiments on a machine with two Xeon Gold 6252 processors (24 cores, 3.7GHz max frequency, 33MB L3 cache, with 2-way hyperthreading).
The machine has 375GB of DRAM and 3TB of NVRAM (Intel Optane DC memory), organized as 12 $\times$ 256GB  DIMMS (6 per processor).

On Intel/AMD architectures~\cite{intel,amd}, the three available flush instructions are \emph{clflush}, \emph{clflushopt}, and \emph{clwb}, where \emph{clwb} is not blocking and supposed to not invalidate the cache. Thus, \emph{clwb} is the most efficient one, and is the one we use in our implementation. The processors are based on the Cascade Lake SP microarchitecture, which supports the \texttt{clwb} instruction for flushing cache lines (\flush{}). However, its implementation of \emph{clwb} still invalidates cache lines.
Performance might be improved in future platforms where clwb does not invalidate cache lines.
For ordering, we use the \emph{sfence} instruction. 
The equivalent instructions on ARM are \emph{DC CVAP} and a full system \emph{DSB} instruction for flush and fence execution~\cite{arm}.
We use \emph{libvmmalloc} from the PMDK library to place all dynamically allocated objects in NVRAM, which is configured in an App-Direct mode \hedit{to let the NVRAM reside alongside the DRAM and allow byte addressable access} \Guy{say a few words about app-direct mode}. All other objects are stored in RAM.
The operating system is Fedora 27 (Server Edition), and the code was written in C++ and compiled using g++ (GCC) version~7.3.1.
We use \texttt{std::atomic}s with relaxed memory orders where appropriate.
In our implementation of Algorithm~\ref{alg:fmalg}, some of the \fence{} instructions can be omitted because on our Intel machine, atomic instructions (such as CAS and FAA) perform an implicit \fence{}.

We avoid crossing NUMA-node boundaries, since unexpected effects have been observed when allocating across NUMA nodes on the NVRAM .
Hyperthreading is used for experiments with more than 24 threads. Unless stated otherwise, all data structures are tested with three different workloads; $0\%$ updates, $5\%$ updates, and $50\%$ updates. Updates are split $50/50$ between inserts and deletes, and chosen randomly.
All experiments were run for 5 seconds and an average of 5 runs is reported. 
A grey dotted line (shown in some plots) represents the original (non-persisted) form of the tested data structure.


%
%
%


\subsection{Testing Hash Table Size}

We begin our evaluation by testing the effect of the size of the \fmhash{} on performance. 
There is a trade-off between the size of the \fmhash{} and the number of collisions on the counters; keeping the table small allows it to fit in cache, making accesses to it potentially cheaper. However, if it is too small, hash collisions could cause cache coherence misses. Figure~\ref{fig:hashtables} shows the result of different \fmhash{} sizes on the BST, with three different update ratios. We show the automatic BST  implementation. Other data structures showed similar patterns, and are omitted for brevity.

We first note that for $0\%$ updates, we see that the larger the hash table, the lower the throughput. This is as expected; as the \fmhash{} grows, less of it fits in cache, and therefore accesses to it more frequently incur cache misses. Furthermore, at $0\%$ updates, the \fmcounter{s} are never updated, so coherence misses are not a concern.
Starting at $5\%$ updates, we see a stark performance drop for the $4$KB hash table. Two types of hash collisions can occur in this framework: (1) two locations hash to the same counter, resulting in potentially redundant \flush{s} executed, if the \fmcounter{} balance is inflated due to an ongoing \flushedm{} on a different location.
More severe, however, is the second type of hash collisions: (2) cache line collisions; the 4K \fmcounter{s} in the hash table are packed into only $64$ cache lines. This means that if any two \flushedi{s}, at least one of which is a \flushedm{}, occur on locations that hash to the same cache line (quite likely), they suffer a coherence cache miss. In such a small hash table, this effect is very prominent.
This is much less noticeable in the larger hash tables.

For the rest of the plots, we show only one hash table size; the 1MB \fmhash. We note that this size fits in the L3 cache, but is large enough to avoid most hash collisions.

\ifx\arxiv\undefined
\begin{figure}
	\includegraphics[width=0.46\textwidth]{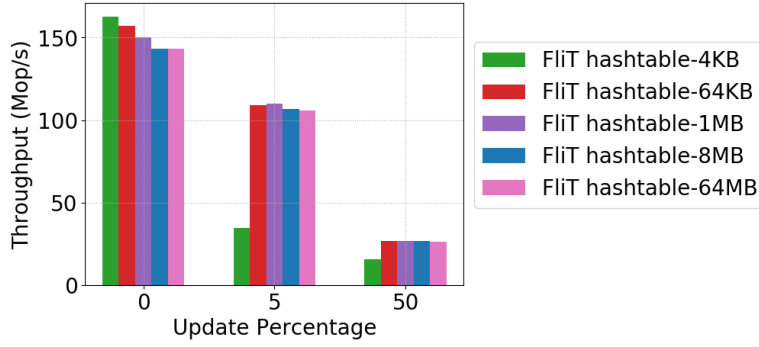}
	\caption{\small Tuning hashtable size for the FliT library. Throughput shown is for the automatic BST with 10K keys.}
	\label{fig:hashtables}
\end{figure}
\else
\begin{figure} \centering
	\includegraphics[width=0.6\textwidth]{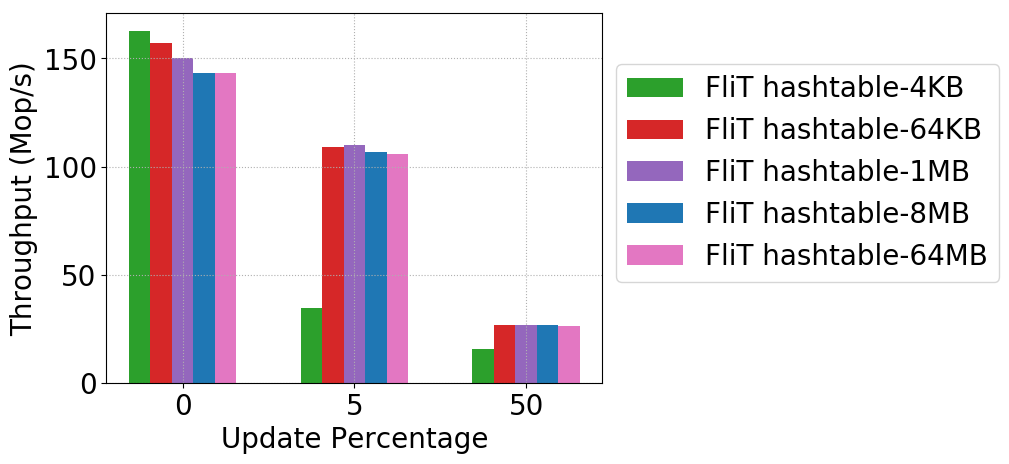}
	\caption{\small Tuning hashtable size for the FliT library. Throughput shown is for the automatic BST with 10K keys.}
	\label{fig:hashtables}
\end{figure}
\fi

\subsection{Varying Number  of Threads}
We now consider the scalability of data structures that use the \FM{} library, as the number of threads grows. The results can be seen in Figure~\ref{fig:scale}. Again, the automatic 10K BST  with $5\%$ updates is shown. Note that in this plot, aside from the \fmhash{}, we show a few different settings for comparison. In particular, the gray line shows a non-persistent version of the data structure, in which no \flush{} or \fence{} instructions are issued. This forms a baseline that cannot be significantly outperformed by any persistent implementation. Furthermore, the blue line shows a BST version implemented with \noflit{} \flush{} and \fence{} usage, without applying  the \FM{} library at all. This version performs many more \flush{s}, and its performance and scalability suffer. We show both the \fmhash{} and and the \fmadjacent{} versions of the \FM{} library. Both of them scale similarly, and quite well.


\ifx\arxiv\undefined
\begin{figure}
	\includegraphics[width=0.46\textwidth]{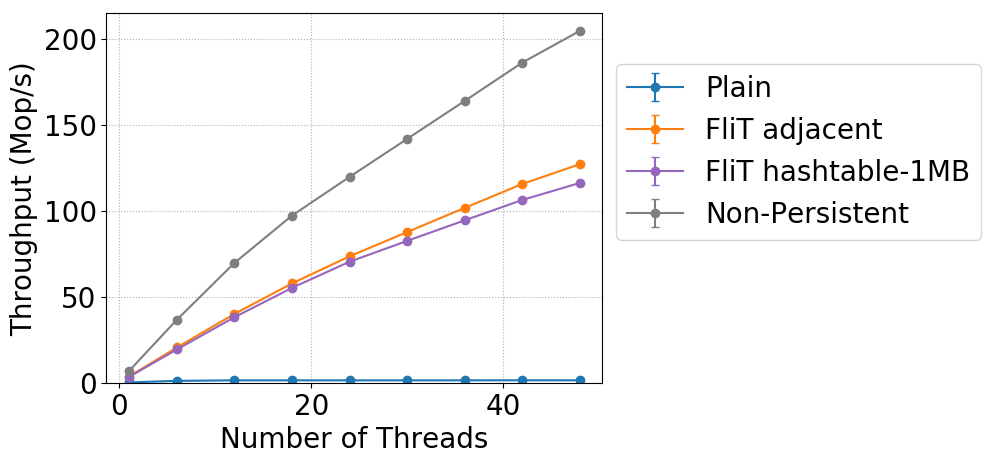}
	\caption{\small Scalability graphs for the automatic BST with 10K keys and 5\% updates.}
	\label{fig:scale}
\end{figure}
\else
\begin{figure} \centering
	\includegraphics[width=0.6\textwidth]{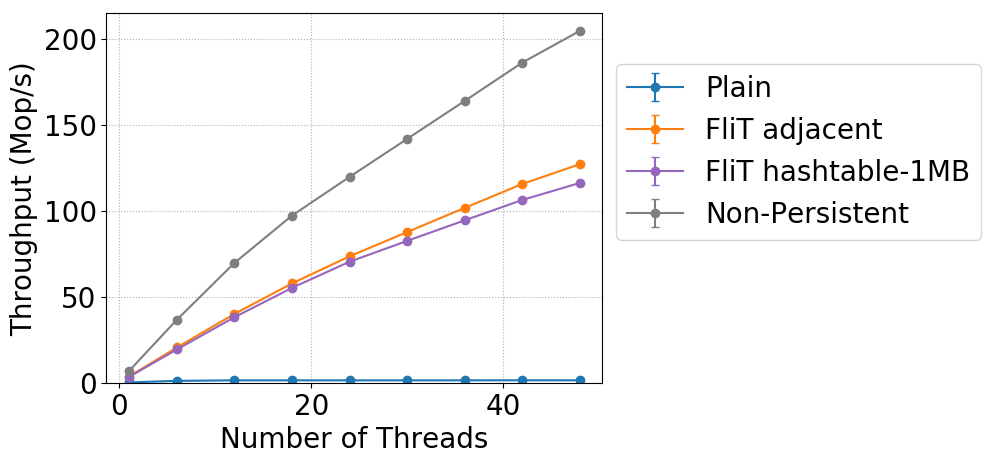}
	\caption{\small Scalability graphs for the automatic BST with 10K keys and 5\% updates.}
	\label{fig:scale}
\end{figure}
\fi

\subsection{Comparing Durability Methods}

Figure~\ref{fig:varyversions} shows the four implemented data structures, each with their three different methods of durability: automatic, NVtraverse, and manual. When using the \FM{} library, these methods differ in how many \nonflushedi{s} they execute; the automatic version only executes \flushedi{s}, the NVtraverse executes many \nonflushedl{s} while traversing the data structure, and the manual version carefully reasons about these individual data structures to make a larger fraction of the instructions be \nonflushed. 
All plots show $5\%$ updates, and the smaller size of the tested data structure (10K nodes for the scalable data structures, and 128 nodes for the linear linked-list). For each setting, we show a \noflit{} implementation, \fmadjacent{}, \fmhash, and link-and-persist where applicable. 

Generally speaking, the link-and-persist method follows the same patterns as the \FM{} implementations.
We note that the more optimized the underlying durability implementation is, the less it benefits from \FM{}. However, for all settings, the performance boost from \FM{} is still substantial; while in the automatic version, \FM{} boosts throughput by  a factor of at least $6.68\times$ (in the hash table), and at most $99.5\times$ (in the skiplist), we still observe an improvement of at least $2.17\times$ when using \FM{} in all data structures under all durability methods. However, it is also important to note that across the board, the optimized durability methods with \FM{} outperform the automatic durability method with \FM. Thus, while benefiting less from the \FM{} library, optimizations that allow using more \nonflushedi{s} are still useful, and should still be implemented using the \FM{} library.

Interestingly, while optimized solutions do perform better, the automatic version implemented with the \FM{} library performs surprisingly well; it significantly outperforms the NVtraverse and manual versions without the \FM{} library for the BST and hash table, and approximately matches their performance in the linked-list and skiplist.

\ifx\arxiv\undefined
\begin{figure*}[!h]
	\begin{tabular}{@{}c@{}c@{}c@{}c@{}c@{}c}
		\multicolumn{5}{l}{\includegraphics[width=0.65\textwidth]{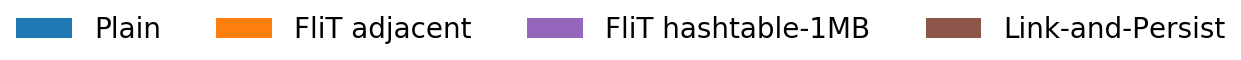}} \\
		&\includegraphics[width=0.25\textwidth]{graphs/throughput:_Bst_10000_keys_5_updates_44_threads.png} 
		&\includegraphics[width=0.25\textwidth]{graphs/throughput:_Hashtable_10000_keys_5_updates_44_threads.png}
		&\includegraphics[width=0.25\textwidth]{graphs/throughput:_List_128_keys_5_updates_44_threads.png}
		&\includegraphics[width=0.25\textwidth]{graphs/throughput:_Skiplist_10000_keys_5_updates_44_threads.png} \\
		&(a) \small BST, 10K keys
		&(b) \small Hashtable, 10K keys
		&(c) \small Linked List, 128 keys
		&(d) \small Skiplist, 10K keys
	\end{tabular}
	\caption{\small Throughput results with 44 threads and 5\% updates. Dotted bar represents throughput of the non-persistent version of each data structure. }
	\label{fig:varyversions}
\end{figure*}
\else
\begin{figure*}[!h]
	\begin{tabular}{@{}c@{}c@{}c@{}c@{}c@{}c}
		\multicolumn{5}{l}{\includegraphics[width=0.65\textwidth]{graphs/throughput_compare_versions_legend.png}} \\
		&\includegraphics[width=0.25\textwidth]{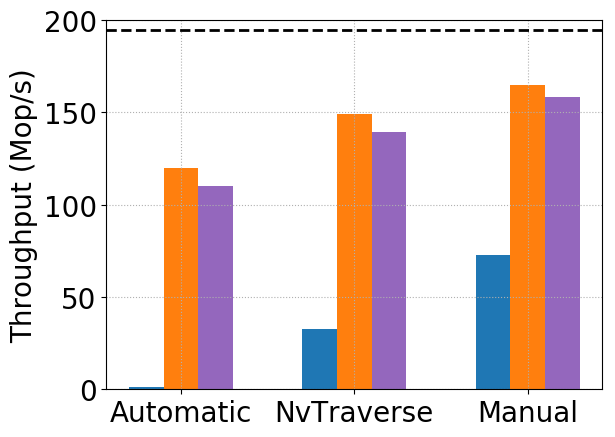} 
		&\includegraphics[width=0.25\textwidth]{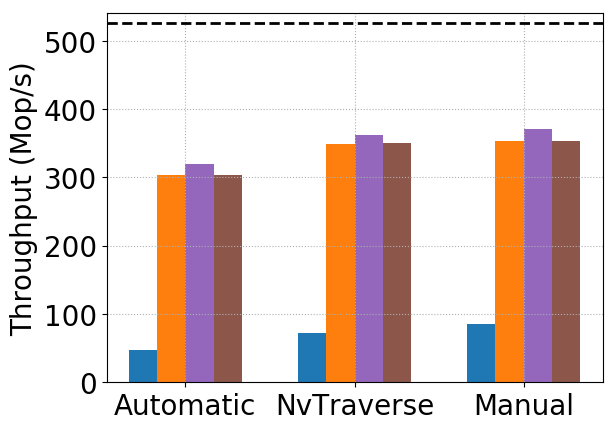}
		&\includegraphics[width=0.25\textwidth]{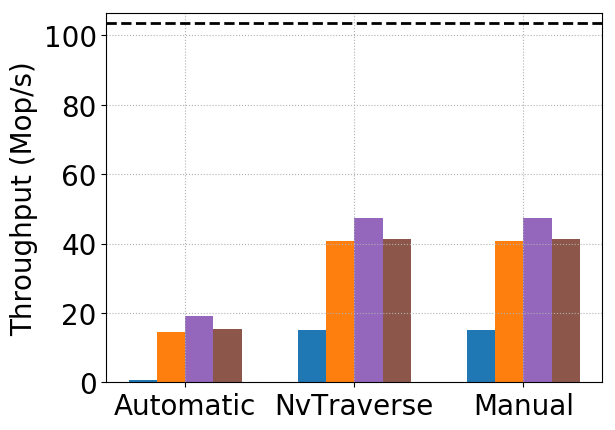}
		&\includegraphics[width=0.25\textwidth]{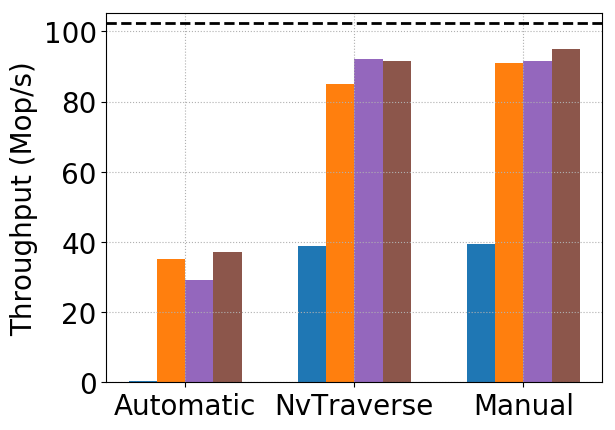} \\
		&(a) \small BST, 10K keys
		&(b) \small Hashtable, 10K keys
		&(c) \small Linked List, 128 keys
		&(d) \small Skiplist, 10K keys
	\end{tabular}
	\caption{\small Throughput results with 44 threads and 5\% updates. Dotted bar represents throughput of the non-persistent version of each data structure. }
	\label{fig:varyversions}
\end{figure*}
\fi

\subsection{Effect of updates.}

In Figure~\ref{fig:varyupdates}, we show each data structure with two different sizes, and in each subplot, we vary the update ratio of the workload. These plots are normalized to the throughput of the non-persistent baseline for each data structure. It is easy to see that the more updates executed, the worse the performance of all persistent versions when compared to the non-persistent baseline. This is expected; \flush{} and \fence{} instructions are executed more the more update operations occur. Note that in $0\%$ update workloads, no \flush{} or \fence{} instructions are executed, other than in initialization and at the end of each operation in the \FM{} and link-and-persist versions, since loads only ever execute a \flush{} if the location is tagged (and only \flushedm{s} can tag memory locations).

Furthermore, in $0\%$ updates, the \fmadjacent{} and the link-and-persist do better than the hash-table variant. This is because the latter implementations never have to incur an extra cache miss to access the \fmcounter{}, whereas the \fmhash{} incurs L2 misses every time it accesses the counter.

\subsection{Comparing \FM{} and Link-and-Persist}
We note that in general, the \fmadjacent{} and the link-and-persist implementations perform almost identically. This is because they both avoid this extra cache miss when accessing the \fmcounter{} (or flush-bit in the case of link-and-persist). The exception to this rule is in the skiplist, where link-and-persist outperforms the \fmadjacent. This is because \fmadjacent{} doubles the size of each node. In most data structures, it goes unnoticed, since each node still fits in a single cache line. However, the skiplist node stores many pointers, and thus can overflow a cache line when each word in it is doubled to fit the \fmcounter. This problem does not occur with the link-and-persist. However, we note that link-and-persist is not as general, and cannot be implemented with the BST, since this BST algorithm makes use of all bits in each word.

Interestingly, while \fmadjacent/link-and-persist perform best when there are $0\%$ updates, this is not always the case when more updates occur. This is most noticeable in the smaller hash table and linked-list implementations (Figures~\ref{fig:varyupdates}b and c). This is because, while in our implementation, we use Intel's \emph{clwb} instruction to perform \flush{s}, which should not invalidate cache lines update flushing, invalidations still occur. Indeed, Intel confirms that clwb, while available for use, is not currently implemented in hardware. Therefore, \flushedm{s} in the \fmadjacent{} incur a cache miss when decrementing the \fmcounter, since they always do so after having executed a clwb on that cache line, thereby evicting it from memory. The same thing happens in the link-and-persist implementation, when flipping the flush-bit after having flushed the cache line. Since the \fmhash{} does not place the \fmcounter{} on the same cache line as its \flushedm{} is accessing, decrementing the \fmcounter{} does not incur this cache miss. This effect is less prominent in larger data structures, in which traversing the data structure dominates the overall execution time. Furthermore, we believe that this effect will disappear once Intel implement their non-invalidating flush option in hardware.

\subsection{Comparing Data Structure Size}

Like the advantage of the \fmhash{} over the \fmadjacent, some patterns observed in the smaller data structures become less prominent when run on a larger instantiation of data structure.  In particular, for large data structures, the persistent versions almost match the performance of the non-persistent baseline. This is because in large data structures, the runtime of an operation is dominated by the need to traverse long chains of nodes, incurring many cache misses (since the large BST, Hash Table, and Skiplist do not fit in cache). The overhead of occasional \flush{s} and \fence{s} becomes less significant.

\ifx\arxiv\undefined
\begin{figure*}[!h]
\begin{tabular}{@{}c@{}c@{}c@{}c@{}c@{}c}
	\multicolumn{5}{l}{\includegraphics[width=0.65\textwidth]{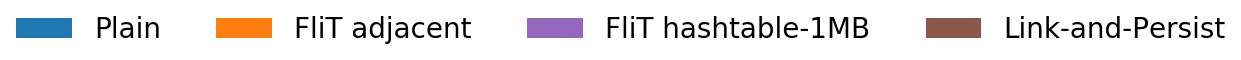}} \\
	&\includegraphics[width=0.25\textwidth]{graphs/throughput:_Bst_10000_keys_automatic_44_threads.png} 
	&\includegraphics[width=0.25\textwidth]{graphs/throughput:_Hashtable_10000_keys_automatic_44_threads.png}
	&\includegraphics[width=0.25\textwidth]{graphs/throughput:_List_128_keys_automatic_44_threads.png} 
	&\includegraphics[width=0.25\textwidth]{graphs/throughput:_Skiplist_10000_keys_automatic_44_threads.png} \\ 
	&(a) \small BST, 10K keys
	&(b) \small Hashtable, 10K keys
	&(c) \small Linked List, 128 keys
	&(d) \small Skiplist, 10K keys \\
	&\includegraphics[width=0.25\textwidth]{graphs/throughput:_Bst_10000000_keys_automatic_44_threads.png}  
	&\includegraphics[width=0.25\textwidth]{graphs/throughput:_Hashtable_10000000_keys_automatic_44_threads.png}
	&\includegraphics[width=0.25\textwidth]{graphs/throughput:_List_4096_keys_automatic_44_threads.png}
	&\includegraphics[width=0.25\textwidth]{graphs/throughput:_Skiplist_1000000_keys_automatic_44_threads.png} \\
	&(a) \small BST, 10M keys
	&(b) \small Hashtable, 10M keys
	&(c) \small Linked List, 4K keys
	&(d) \small Skiplist, 10M keys
\end{tabular}
\caption{\small Throughput results for 44 threads, automatic, normalized to the throughput of the non-persistent version of each data structure.}
\label{fig:varyupdates}
\end{figure*}
\else
\begin{figure*}[!h]
	\begin{tabular}{@{}c@{}c@{}c@{}c@{}c@{}c}
		\multicolumn{5}{l}{\includegraphics[width=0.65\textwidth]{graphs/throughput_compare_updates_legend.png}} \\
		&\includegraphics[width=0.25\textwidth]{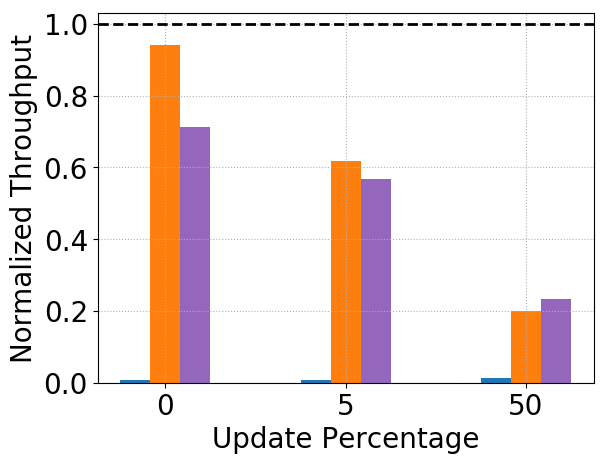} 
		&\includegraphics[width=0.25\textwidth]{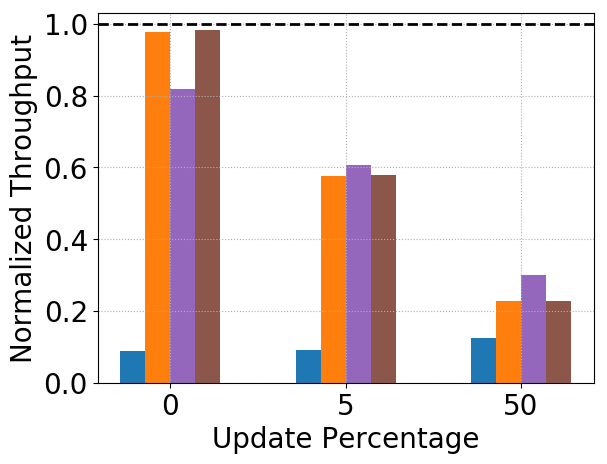}
		&\includegraphics[width=0.25\textwidth]{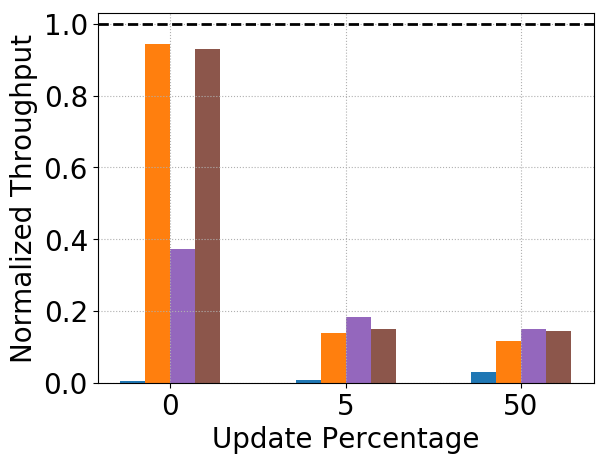} 
		&\includegraphics[width=0.25\textwidth]{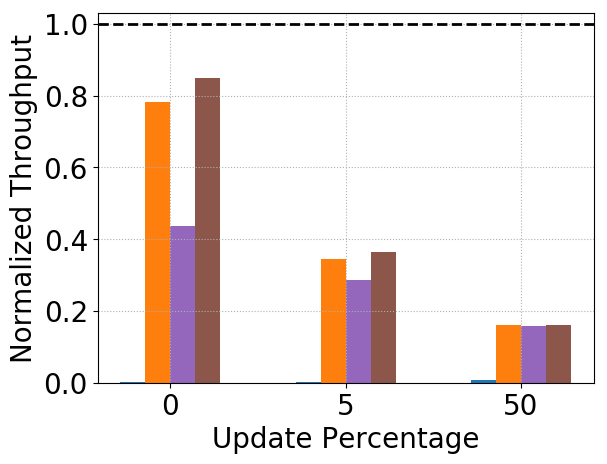} \\ 
		&(a) \small BST, 10K keys
		&(b) \small Hashtable, 10K keys
		&(c) \small Linked List, 128 keys
		&(d) \small Skiplist, 10K keys \\
		&\includegraphics[width=0.25\textwidth]{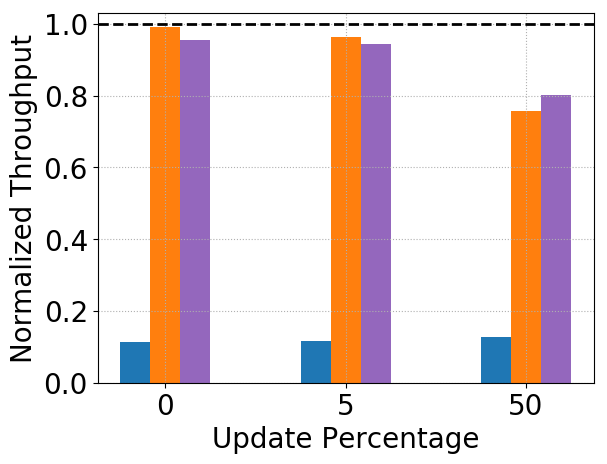}  
		&\includegraphics[width=0.25\textwidth]{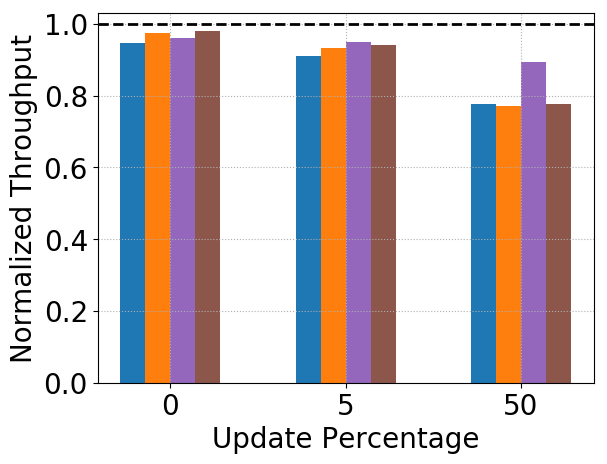}
		&\includegraphics[width=0.25\textwidth]{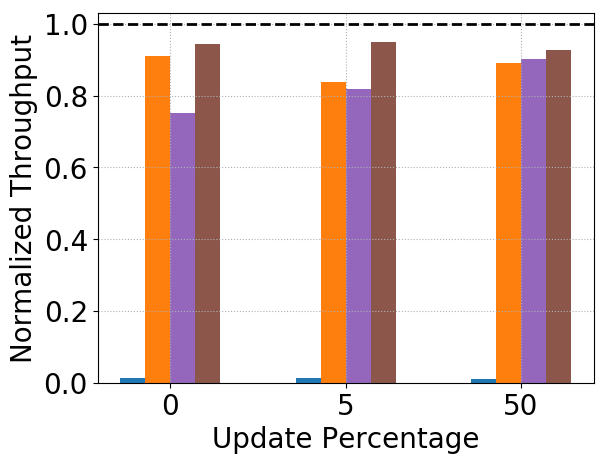}
		&\includegraphics[width=0.25\textwidth]{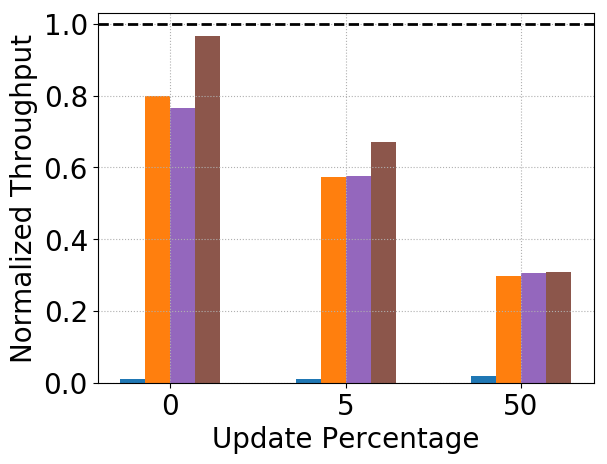} \\
		&(a) \small BST, 10M keys
		&(b) \small Hashtable, 10M keys
		&(c) \small Linked List, 4K keys
		&(d) \small Skiplist, 10M keys
	\end{tabular}
	\caption{\small Throughput results for 44 threads, automatic, normalized to the throughput of the non-persistent version of each data structure.}
	\label{fig:varyupdates}
\end{figure*}
\fi

\subsection{Number of \flush{s}}

Finally, we also measure the number of \flush{} instructions executed in each of our runs. Figure~\ref{fig:flushes} shows some representative results.
In general, the number  of \flush{s} executed is approximately the same across the different \FM{} implementations. This shows that redundant \flush{s}, resulting from a memory location still being tagged when a \flushedl{} accesses it, almost never occur. The exception is seen in the linked-list automatic durability setting, where \fmadjacent{} and link-and-persist perform significantly more \flush{s} than \fmhash. This again confirms the observation that the clwb instruction does invalidate cache lines; since accessing the \fmcounter{} for decrementing it incurs a cache miss in \fmadjacent, the location remains tagged for a longer period of time. In the automatic setting, in which all loads are \flushedl{s}, this means that it is not uncommon for another thread to read the \fmcounter{} while the location is tagged, causing an extra \flush.
 
\ifx\arxiv\undefined
\begin{figure}
	\begin{tabular}{ccc}
		\multicolumn{3}{c}{\includegraphics[width=0.4\textwidth]{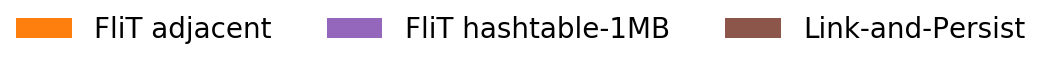}} \\
		&\includegraphics[width=0.22\textwidth]{graphs/flushes:_Hashtable_10000_keys_5_updates_44_threads.png} 
		&\includegraphics[width=0.22\textwidth]{graphs/flushes:_List_128_keys_5_updates_44_threads.png} \\ 
		&(a) \small Hashtable, 10K keys
		&(b) \small List, 128 keys
	\end{tabular}
	\caption{\small Number of flushes per operation, $5\%$ updates}
	\label{fig:flushes}
\end{figure}
\else
\begin{figure} \centering
	\begin{tabular}{ccc}
		\multicolumn{3}{c}{\includegraphics[width=0.6\textwidth]{graphs/flushes_compare_versions_legend.png}} \\
		&\includegraphics[width=0.22\textwidth]{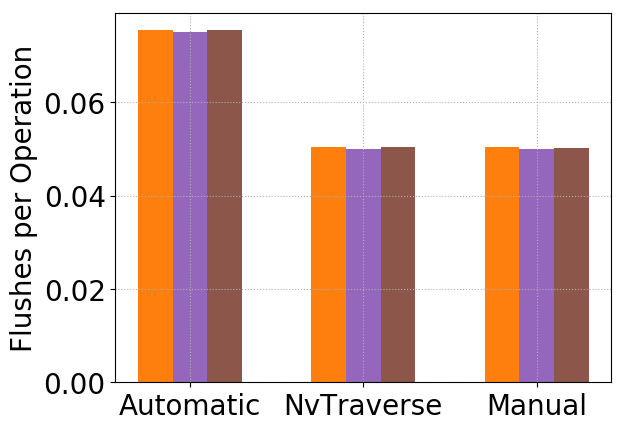} 
		&\includegraphics[width=0.22\textwidth]{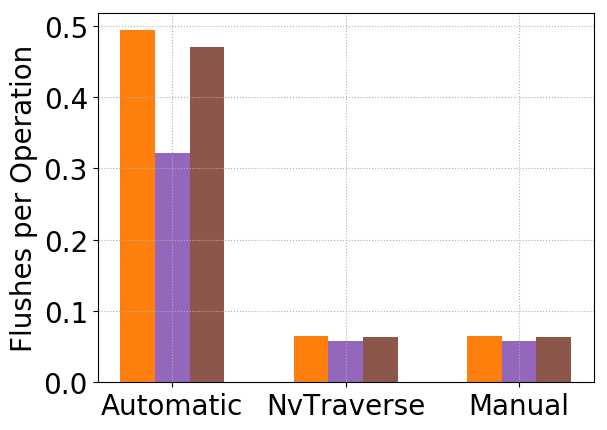} \\ 
		&(a) \small Hashtable, 10K keys
		&(b) \small List, 128 keys
	\end{tabular}
	\caption{\small Number of flushes per operation, $5\%$ updates}
	\label{fig:flushes}
\end{figure}
\fi

	\section{Related Work}\label{sec:related}

There have been many papers focusing on finding how to easily and efficiently program for NVRAM. 
Izraelevitz et al.~\cite{izraelevitz2016linearizability} present the notion of \emph{durable linearizability}, a correctness condition for persistent data structures. At a high level, durable linearizability requires a data structure to be linearizable despite any number of system crashes that occur during its execution. Izraelevitz el al. also showed how to place \flush{} and \fence{} instructions in linearizable code with acquire-release consistency to guarantee durable linearizability. In this paper we show how to rewrite this construction in terms of the \interface, and use the \FM{} library to optimize this implementation. Izraelevitz et al also introduce a weaker correctness guarantee, called buffered durable linearizability, which we do not consider in this paper. 

Researchers have also introduced other correctness criteria for persistence and explored how to support them efficiently~\cite{bendavid2019delay,friedman2018persistent,berryhill2016robust,attiya2018nesting,attiya2020tracking}.
These works consider not only the state of shared memory upon recovery from a system crash, but also whether processes can continue their previous execution. For example, detectability~\cite{friedman2018persistent}, requires that each process be able  to find out whether its most recently called operation had completed before a crash. These conditions can be achieved by storing extra metadata beyond what is stored in a non-persisted execution. We believe that using the \interface{} when designing algorithms for these other correctness criteria can improve performance and portability just as much as it  does for durably linearizable implementations.

Many algorithms have been designed for NVRAM in the context of file systems and database indexes~\cite{chen2015persistent,lee2017wort,lejsek2009nv,xu2016nova,venkataraman2011consistent,yang2015nv,lee2019recipe}. These algorithms are often lock-based, rather than the lock-free data structures that we have compared to in our evaluation. We believe that the \interface{}, and its implementation in the \FM{} library, can also be used to enhance such algorithms. Indeed, Lee et al. use a technique (like link-and-persist~\cite{david2018logfree}) in their B-tree algorithm~\cite{lee2019recipe}. However, we focused on lock-free data structures in our evaluation since the largest benefits in the \FM{} library's implementation can be seen in contended workloads, which are less prominent in lock-based algorithms. Still, the \interface{} captures lock-based algorithms as well, leaving room for optimized solutions by treating private instructions (those inside a lock) separately from shared instructions. Similarly, we believe the \interface{} can be used to write and reason about efficient persistent transactional memories, a topic that has also drawn significant attention in recent years~\cite{correia2018romulus,ramalhete19onefile,pmdk,beadle2020nonblocking}.

Several papers provide other programming interfaces for NVRAM. Mnemosyne~\cite{volos2011mnemosyne} provides an interface for using persistent memory through \emph{persistent regions}. 
Atlas~\cite{chakrabarti2014atlas} provides persistence for general lock-based programs, but does not capture lock-free algorithms. 
Gogte et al~\cite{gogte2018persistency} propose semantics for persistent synchronization-free regions.
Other works capture the persistence semantics offered by modern architectures, like Intel-X86~\cite{raad2019persistency} and ARMv8~\cite{raad2019model}. This line of work differs from ours in its goals; we propose an interface for easy persistent programming, which can be implemented in hardware using the semantics formalized in these papers.
%
	\section{Conclusion}\label{sec:conclusion}

In this paper, we introduce the \FM{} library, a C++ library for designing simple and efficient persistent programs for NVRAM. \FM{} is implemented in a way that avoids unnecessary flushing by using \emph{\fmcounter{s}} to track dirty cache lines.

We test \FM{} on an Intel machine with Optane DC memory, and demonstrate that the \FM{} library not only achieves remarkable speedups over even the most optimized persistent data structures, but is also widely applicable. 

Our implementation tested two different ways of allocating the \fmcounter{s} and mapping them to memory locations. Many other variants are possible, and it would be interesting to see the effects of different counter allocation strategies on algorithms that use \FM. One natural option that we did not explore is to assign one counter per cache line rather than at the granularity of words. 

While \FM's default mode makes any linearizable data structure durable with minimal code changes and impressive performance, it also allows further optimizations. In particular, it allows a programmer to specify some instructions that do not need to be persisted. We capture this flexibility with the \interface, which defines the semantics of code in which some memory instructions can remain volatile, while others must be persisted. This interface is language- and architecture-agnostic, and we show it captures persistence behavior in many algorithms; we believe that the \interface{} can be implemented on different architectures, and would achieve similar performance gains as those achieved by \FM.
Note that the algorithm for maintaining \fmcounter{s} is more general than the \interface{} and can be used to implement other persistent interfaces as well.

	\bibliographystyle{plain}
	\bibliography{strings,biblio}

\end{document}